\documentclass[11pt,english]{article}
\usepackage[latin9]{inputenc}
\usepackage[a4paper]{geometry}
\usepackage{color}
\usepackage{amsmath}
\usepackage{amsthm}
\usepackage{amssymb}
\usepackage{amsmath}
\usepackage{mathcomp}
\usepackage{dsfont}
\usepackage[toc,page]{appendix}
\usepackage{graphicx} 
\usepackage{subfigure}
\usepackage{enumerate}
\usepackage[dvipsnames]{xcolor}
\usepackage{physics}
\usepackage{comment}

\usepackage{hyperref}

\newcommand{\Gfun}{\mathcal{G}}
\newcommand{\Efun}{\mathcal{E}}
\newcommand{\Ffun}{\mathcal{F}}
\newcommand{\Rthree}{\mathbb{R}^3}
\newcommand{\infspec}{\text{inf\,spec\,}}

\newcommand\R{\mathbb{R}}
\newcommand\eps{\varepsilon}

\DeclareMathOperator{\dist}{dist}

\newtheorem{definition}{Definition}[section]
\newtheorem{lemma}[definition]{Lemma}

\newtheorem{theorem}[definition]{Theorem}
\newtheorem{corollary}[definition]{Corollary}
\theoremstyle{definition}
\newtheorem{remark}[definition]{Remark}

\numberwithin{equation}{section}

\def\Re{\mathrm{Re}}
\def\Im{\mathrm{Im}}

\title{Persistence of the spectral gap for the Landau--Pekar equations}
\author{Dario Feliciangeli, Simone Rademacher and Robert Seiringer }

\begin{document}
\maketitle 

\begin{abstract}
The Landau--Pekar equations  describe the dynamics of a strongly coupled polaron. Here we provide a class of initial data for which the associated effective Hamiltonian has a uniform spectral gap for all times. For such initial data, this allows us to extend the results on the adiabatic theorem for the Landau--Pekar equations and their derivation from the Fr\"ohlich model obtained in \cite{LRSS,LMRSS} to larger times. 
\end{abstract}

\section{Introduction and Main Results}
\label{section1}

The Landau--Pekar equations \cite{landau} provide an effective description of the dynamics for a strongly coupled polaron, modeling an electron moving in an ionic crystal. The strength of the interaction of the electron with its self-induced polarization field is described by a coupling parameter $\alpha>0$. In this system of coupled differential equations, the time evolution of the electron wave function $\psi_t \in H^1( \mathbb{R}^3)$ is governed by a Schr\"odinger equation with respect to an effective Hamiltonian $h_{\varphi_t}$ depending on the polarization field  $\varphi_t \in L^2( \mathbb{R}^3)$, which evolves according to a classical field equation. Motivated by the recent work in \cite{LRSS,M,LMRSS}, we are interested in initial data for which the Hamiltonian $h_{\varphi_t}$ possesses a uniform  spectral gap (independent of $t$ and $\alpha$) above the infimum of its spectrum. 

The Landau--Pekar equations are of the form
\begin{equation}\label{eq:LP}
\begin{aligned}
 i \partial_t \psi_t  & = h_{ \varphi_t  }  \psi_t  \\
i \alpha^2 \partial_t \varphi_t  & = \varphi_t + \sigma_{\psi_t}
\end{aligned}
\end{equation}
with
\begin{equation}
h_{\varphi} = - \Delta  + V_\varphi, \quad V_\varphi(x) = 2 (2\pi)^{3/2} \Re\, [ ( - \Delta)^{-1/2} \varphi ] (x), \quad \sigma_{\psi} (x) = (2\pi)^{3/2}\left[ (- \Delta)^{-1/2} \vert \psi \vert^2\right]  (x) .
\end{equation}
For initial data $( \psi_0, \varphi_0) \in H^1( \mathbb{R}^3 ) \times L^2( \mathbb{R}^3) $, \eqref{eq:LP} 
is well-posed for all times $t \in \mathbb{R}$ (see \cite{FG} or Lemma \ref{lemma:wellposedness} below). 

For $(\psi,\varphi)\in H^1(\Rthree)\times L^2(\Rthree)$ with $\|\psi\|_2=1$, the energy functional corresponding to the Landau--Pekar equations is defined as
\begin{equation}\label{def:G}
\mathcal{G} ( \psi, \varphi ) = \langle \psi, h_\varphi \psi \rangle + \| \varphi \|_2^2.
\end{equation}
One readily checks that for solutions of \eqref{eq:LP}, $\mathcal{G}(\psi_t,\varphi_t)$ is independent of $t$ \cite[Lemma 2.1]{FG}, and the same holds for $\|\psi_t\|_2$.  
We also define 
\begin{equation}\label{def:EF}
\Efun(\psi)=\inf_{\varphi \in L^2(\Rthree)} \Gfun(\psi,\varphi), \quad \quad \Ffun(\varphi)=\inf_{\psi \in H^1(\Rthree) \atop \|\psi\|_2=1 } \Gfun(\psi,\varphi).
\end{equation}
These three functionals are known as Pekar functionals and we shall discuss some of their properties in Section~\ref{section2}. It follows from the work in \cite{L} that there exist $( \psi_{\mathrm{P}}, \varphi_{\mathrm{P}}) \in H^1( \mathbb{R}^3) \times L^2( \mathbb{R}^3)$ with $\|\psi_{\mathrm{P}}\|_2=1$, called Pekar minimizers, realizing 
\begin{equation}
\inf_{\psi,\varphi} \Gfun(\psi,\varphi) =  \Gfun ( \psi_{\mathrm{P}}, \varphi_{\mathrm{P}})=\Efun(\psi_{\mathrm{P}})=\Ffun(\varphi_{\mathrm{P}}) = e_{\mathrm{P}} < 0 \,,
\end{equation}
and $( \psi_{\mathrm{P}}, \varphi_{\mathrm{P}})$ is unique up to symmetries (i.e., translations and multiplication of $\psi_{\mathrm{P}}$ by a constant phase factor). 
We also note that the Hamiltonian $h_{\varphi_{\mathrm{P}}}$ has a spectral gap above its ground state energy, i.e., $\Lambda(\varphi_{\mathrm{P}}) >0$, where we denote for general $\varphi \in L^2(\R^3)$
\begin{equation}\label{def:gap}
\Lambda ( \varphi) = \inf_{\substack{\lambda \in \mathrm{spec}( h_{\varphi}) \\ \lambda \not= e(\varphi) }} \vert \lambda - e( \varphi) \vert  \quad \text{with} \quad e(\varphi) = \infspec h_\varphi \,.
\end{equation}

In the following we consider solutions $(\psi_t, \varphi_t)$ to the Landau--Pekar equations \eqref{eq:LP} with initial data $( \psi_0, \varphi_0)$ such that its energy $\Gfun(\psi_0,\varphi_0)$ is sufficiently close to $e_{\mathrm{P}}$, and show that for such initial data 
the Hamiltonian $h_{\varphi_t}$ possesses a uniform spectral gap above the infimum of its spectrum for all times $t \in \mathbb{R}$ and any coupling constant $\alpha>0$. This is  the content of the following Theorem.

\begin{theorem}
\label{thm:specgap}
For any $0 < \Lambda < \Lambda ( \varphi_{\mathrm{P}})$ there exists $\varepsilon_\Lambda >0$ such that if  $( \psi_t, \varphi_t )$ is the solution of the Landau--Pekar equations \eqref{eq:LP} with initial data $(\psi_0, \varphi_0) \in H^1( \mathbb{R}^3) \times L^2( \mathbb{R}^3)$  with $\| \psi_0 \|_2=1$  and $\mathcal{G} (\psi_0, \varphi_0) \leq e_{\mathrm{P}} + \varepsilon_\Lambda$, then
\begin{equation}
\label{ass:1}
\Lambda (\varphi_t )\geq \Lambda \quad \mathrm{for \, all}  \quad t\in \mathbb{R}, \, \alpha >0.  
\end{equation}
\end{theorem}

Theorem \ref{thm:specgap} is proved in Section~\ref{section3}. It provides a class of initial data for the Landau--Pekar equations for which the Hamiltonian $h_{\varphi_t}$ has a uniform spectral gap for all times $t \in \mathbb{R}$. The existence of initial data with this particular property is of relevance for recent work \cite{LRSS,M,LMRSS} on the adiabatic theorem for the Landau--Pekar equations, and  on their derivation from the Fr\"ohlich model (where the polarization is described as a quantum field instead). For this particular initial data, the results  obtained there can then be extended in the following way:

\paragraph{Adiabatic theorem.} Due to the separation of time scales in \eqref{eq:LP}, the Landau--Pekar equations decouple adiabatically for large $\alpha$ (see \cite{LRSS} or also \cite{FG_a} for an analogous one-dimensional model). To be more precise, in \cite{LRSS} the initial phonon state function is assumed to satisfy
\begin{equation}
\label{eq:ass_old}
\varphi_0 \in L^2( \mathbb{R}^3) \quad \mathrm{with} \quad e( \varphi_0) = \infspec  h_{\varphi_0} <0,
\end{equation}
which implies that $h_{\varphi_0}$ has a spectral gap and that there exists a unique positive and normalized ground state $\psi_{\varphi_0}$ of $h_{\varphi_0}$. Under this assumption, denoting by $( \psi_t, \varphi_t) $ the solution of the Landau--Pekar equations \eqref{eq:LP} with initial data $(\psi_{\varphi_0}, \varphi_0)$, \cite[Thm.~II.1 \& Rem.~II.3]{LRSS} proves that there exist constants $C,T>0$ (depending on $\varphi_0$) such that 
\begin{equation}
\label{eq:adiab_old}
\| \psi_t - e^{-i \int_0^t ds \, e( \varphi_s ) } \psi_{\varphi_t} \|_2^2 \leq C \alpha^{-4}   \quad \text{for all} \quad |t| \leq T \alpha^2 ,
\end{equation}
where $\psi_{\varphi_t}$ denotes the unique positive and normalized ground state of $h_{\varphi_t}$. The restriction on $|t|$ in \eqref{eq:adiab_old} is due to the need of ensuring that the spectral gap of the effective Hamiltonian $h_{\varphi_t}$ does not become too small for initial data satisfying \eqref{eq:ass_old}, which is only proven (in \cite[Lemma II.1]{LRSS})   for times $|t| \leq T \alpha^2$. Nevertheless, {\em assuming} that there exists $\Lambda >0$ such that $\Lambda ( \varphi_t ) > \Lambda$ for all times $t \in \mathbb{R}$, the adiabatic theorem in \cite[Thm.~II.1]{LRSS} allows to approximate $\psi_t$ by $e^{-i \int_0^t ds \, e( \varphi_s ) } \psi_{\varphi_t} $ for all  times $|t| \ll \alpha^4$. This raises the question about initial data for which the existence of a spectral gap of order one holds true for longer times, and Theorem \ref{thm:specgap} answers this question. In fact, by suitably adjusting the phase factor, we can prove the following  stronger result.

\begin{corollary}
\label{cor:adiab}
Let $\varphi_0\in L^2(\Rthree)$ be such that
\begin{equation}
\label{ass2}
\Ffun(\varphi_0) \leq e_{\mathrm{P}} + \varepsilon
\end{equation}
for  sufficiently small $\varepsilon>0$. Then $h_{\varphi_0}$ has a ground state $\psi_{\varphi_0}$. Let $( \psi_t, \varphi_t)$ be the solution to the Landau--Pekar equations \eqref{eq:LP} with initial data $( \psi_{\varphi_0}, \varphi_0) $ and define 
\begin{equation}
\label{eq:phase_add}
\nu (s) = - \alpha^{-4} \langle \psi_{\varphi_s}, \, V_{\Im\, \varphi_s} R_{\varphi_s}^3 V_{\Im\, \varphi_s} \psi_{\varphi_s} \rangle \quad \text {and} \quad\widetilde{\psi}_t=e^{i\int_0^t ds(e(\varphi_s)+\nu(s))}\psi_t,
\end{equation}
where $R_{\varphi_s} = q_s ( h_{\varphi_s} - e( \varphi_s))^{-1} q_s$ with $q_s = 1 - \vert \psi_{\varphi_s} \rangle \langle \psi_{\varphi_s} \vert$. Then,  there exists a $C>0$ (independent of $\varphi_0$ and $\alpha$) such that 
\begin{equation}\label{eq:adiab_comb}
\| \widetilde{\psi}_t - \psi_{\varphi_t} \|_2^2 \leq C \varepsilon \alpha^{-4}   \left( 1 +  \alpha^{-2} |t|\right) e^{C \alpha^{-4} |t|}  \,.
\end{equation}
\end{corollary}

Our proof in Section \ref{section3} shows  that the smallness condition on $\varepsilon$ in Corollary \ref{cor:adiab} can be made explicit in terms of properties of $\varphi_{\mathrm{P}}$. It also shows that $\min_{\theta \in [0,2\pi)} \|  e^{i\theta} {\psi}_t - \psi_{\varphi_t} \|_2^2 \leq C \eps$ for all times $t$, independently of $\alpha$. The bound \eqref{eq:adiab_comb} improves upon this for large $\alpha$ as long as $\alpha^{-4} |t| e^{C \alpha^{-4} |t|} \ll \alpha^2$ and hence, in particular, for $|t| \lesssim \alpha^4$.  

\paragraph{Effective dynamics for the Fr\"ohlich Hamiltonian.} As already mentioned, the Landau--Pekar equations provide an effective description of the dynamics for a strongly coupled polaron. Its true dynamics is described by the Fr\"ohlich Hamiltonian \cite{frohlich} $H_\alpha$ acting on $L^2( \mathbb{R}^3) \otimes \mathcal{F}$, the tensor product of the Hilbert space $L^2( \mathbb{R}^3)$ for the electron and the bosonic Fock space $\mathcal{F}$ for the phonons. We refer to \cite{LRSS,LMRSS} for a detailed definition. Pekar product states of the form $\psi_t \otimes W( \alpha^2 \varphi_t ) \Omega$, with $(\psi_t,\varphi_t)$  a solution of the Landau--Pekar equations, $W$ the Weyl operator and $\Omega$ the Fock space vacuum,  were proven in \cite[Thm.~II.2]{LRSS} to approximate the dynamics defined by the Fr\"ohlich Hamiltonian $H_\alpha$ for times $|t| \ll \alpha^2$. Recently, it was shown in \cite{LMRSS}  that in order to obtain a norm approximation valid for times of order $\alpha^2$, one needs to implement correlations among phonons, which are captured by a suitable Bogoliubov dynamics  acting on the Fock space of the phonons only. In fact, considering initial data satisfying \eqref{eq:ass_old}, \cite[Theorem I.3]{LMRSS} proves that there exist constants $C,T>0$ (depending on $\varphi_0$) such that 
\begin{equation}
\label{eq:norm_approx1}
\| e^{-iH_\alpha t} \psi_{\varphi_0} \otimes W( \alpha^2 \varphi_0) \Omega -e^{-i \int_0^t ds \, \omega (s)} \psi_t \otimes  W( \alpha^2 \varphi_t) \Upsilon_t  \|_{L^2( \mathbb{R}^3) \otimes \mathcal{F}} \leq C \alpha^{-1} \quad \mathrm{for \, all } \quad |t| \leq T \alpha^2 ,
 \end{equation}
 where $\omega (s) = \alpha^2 \Im \langle \varphi_s , \partial_s \varphi_s\rangle + \| \varphi_s\|_2^2$ and $\Upsilon_t$ is the solution of the dynamics of a suitable  Bogoliubov Hamiltonian on $\mathcal{F}$   (see \cite[Definition I.2]{LMRSS} for a precise definition).  As for the adiabatic theorem discussed above, the restriction to  times $|t| \leq T \alpha^2$ results from the need of a spectral gap of $h_{\varphi_t}$ of order one (compare with \cite[Remark I.4]{LMRSS}), which under the sole assumption \eqref{eq:ass_old} is guaranteed by  \cite[Lemma II.1]{LRSS} only  for  $|t| \leq T\alpha^2$. Theorem \ref{thm:specgap} now provides a class of initial data for which the above norm approximation holds true for all times of order $\alpha^2$, in the following sense.
 
\begin{corollary}
\label{cor:normapprox}
Let $\varphi_0\in L^2(\Rthree)$ be such that
\begin{equation}
\label{eq:ass2}
\Ffun(\varphi_0) \leq e_{\mathrm{P}} + \varepsilon
\end{equation}
for  sufficiently small $\varepsilon>0$.  Then $h_{\varphi_0}$ has a ground state $\psi_{\varphi_0}$. Let $( \psi_t, \varphi_t)$ be the solution to the Landau--Pekar equations \eqref{eq:LP} with initial data $( \psi_{\varphi_0}, \varphi_0)$. Then  there exists a $C>0$ (independent of $\varphi_0$ and $\alpha$) such that 
\begin{equation}
\label{eq:norm_approx2}
\| e^{-iH_\alpha t} \psi_{\varphi_0} \otimes W( \alpha^2 \varphi_0) \Omega -e^{-i \int_0^t ds \, \omega (s)} \psi_t \otimes  W( \alpha^2 \varphi_t) \Upsilon_t   \|_{L^2( \mathbb{R}^3) \otimes \mathcal{F}} \leq C \alpha^{-1} e^{ C \alpha^{-2} |t| } \,.
 \end{equation}
 \end{corollary}

Again, the smallness condition on $\varepsilon$ in Corollary \ref{cor:normapprox} can be made explicit in terms of properties  of $\varphi_{\mathrm{P}}$. 
Corollary \ref{cor:normapprox} is an immediate consequence of Theorem~\ref{thm:specgap} and the method of proof in \cite{LMRSS}, as explained in \cite[Remark I.4]{LMRSS}.

\section{Properties of the Spectral Gap and the Pekar Functionals}
\label{section2}

Throughout the paper, we use the symbol $C$ for generic constants, and their value might change from one occurrence to the next. 

\subsection{Preliminary Lemmas}

We begin by stating some preliminary Lemmas we shall need throughout the following discussion. 

\begin{lemma}[Lemma~2.1 in \cite{FG}] 
	\label{lemma:wellposedness}
	For any $( \psi_0, \varphi_0 ) \in H^1( \mathbb{R}^3) \times L^2( \mathbb{R}^3)$, there is a unique global solution $( \psi_t, \varphi_t) $ of the Landau--Pekar equations \eqref{eq:LP}. Moreover, $\| \psi_0 \|_2 = \| \psi_t \|_2$, $\mathcal{G} ( \psi_0, \varphi_0) = \mathcal{G} ( \psi_t, \varphi_t )$ for all $t \in \mathbb{R}$ and there exists a constant $C>0$ such that 
	\begin{equation}
	\| \psi_t \|_{H^1( \mathbb{R}^3)} \leq C, \quad \| \varphi_t \|_2 \leq C 
	\end{equation}
	for all $\alpha >0$ and all $t \in \mathbb{R}$. 
\end{lemma}

The following Lemma collects some properties of $V_\varphi$ and $\sigma_\psi$ (see also \cite[Lemma III.2]{LRSS} and \cite[Lemma II.2]{LMRSS}).

\begin{lemma} 
	\label{lemma:potential}
	There exists $C>0$ such that for every $\varphi \in L^2( \mathbb{R}^3)$ and $\psi \in H^1(\mathbb{R}^3)$
	\begin{equation}
	\| V_{\varphi} \|_6 \leq C \| \varphi \|_2 \quad \mathit{and} \quad \| V_\varphi \psi \|_2  \leq C \| \varphi \|_2 \| \psi \|_{H^1( \mathbb{R}^3)} .
	\end{equation}
	Moreover, there exists $C>0$ such that for all $\psi_1, \psi_2 \in H^1(\mathbb{R}^3)$
	\begin{equation}
	\label{eq:sigmaH1bounds}
	\| \sigma_{\psi_1} - \sigma_{\psi_2} \|_2 \leq C \left( \| \psi_1 \|_{2} + \| \psi_2 \|_{2} \right) \min_{\theta\in[0,2\pi)}\|e^{i \theta} \psi_1 - \psi_2 \|_{H^1( \mathbb{R}^3)}.
	\end{equation} 
\end{lemma}

\begin{proof}
	 The first two inequalities follow immediately from \cite[Lemma III.2]{LRSS} and \cite[Lemma II.2]{LMRSS}. For the last inequality, we note that $\sigma_{\psi} = \sigma_{e^{i \theta} \psi}$ for arbitrary $\theta \in \mathbb{R}$. Hence, it is enough to prove the result for $\theta =0$. We write the difference
	\begin{align}
	\widehat{\sigma}_{\psi_1} (k) - \widehat{\sigma}_{\psi_2}  (k) & = |k|^{-1} \left( \langle \psi_1, e^{-ik \,\cdot \,} \psi_1 \rangle -  \langle \psi_2, e^{-ik \, \cdot \,} \psi_2 \rangle\right) \notag \\
	&= |k|^{-1} \left( \langle \psi_1 - \psi_2, e^{-ik \, \cdot\,} \psi_1 \rangle +  \langle \psi_2, e^{-ik \, \cdot\, }\left( \psi_1 - \psi_2\right) \rangle\right) .
	\end{align}
	where $\widehat{\sigma} _\psi (k) = (2 \pi )^{-3/2} \int dx \; e^{-ik \cdot x} \sigma_\psi (x) $ denotes the Fourier transform of $\sigma_\psi$. Thus, 
	\begin{equation}
	\label{eq:diff_sigma_1}
	\| \sigma_{\psi_1} - \sigma_{\psi_2} \|_2^2 \leq 2 \int dk \frac{1}{|k|^2} \left( \vert \langle \psi_1 - \psi_2, e^{-ik\, \cdot \, } \psi_1 \rangle \vert^2 + \vert \langle \psi_2, e^{-ik \,\cdot\, }\left( \psi_1 - \psi_2\right) \rangle\vert^2 \right) . 
	\end{equation}
	For the first term, we write 
	\begin{equation}
	\int \frac{dk}{|k|^2} \vert \langle \psi_1 - \psi_2, e^{-ik \,\cdot\,} \psi_1 \rangle \vert^2 = C \int \frac{dx\,dy}{|x-y|} ( \psi_1 - \psi_2) (x) \overline{(\psi_1 - \psi_2 ) (y)} \, \overline{\psi_1(x)} \psi_1 (y) .
	\end{equation}
	The Hardy-Littlewood-Sobolev inequality implies that
	\begin{equation}
	\int \frac{dk}{|k|^2} \vert \langle \psi_1 - \psi_2, e^{-ik \,\cdot\,} \psi_1 \rangle \vert^2 \leq C \| \psi_1 \overline{( \psi_1 - \psi_2)} \|_{6/5}^2 \leq C \| \psi_1 - \psi_2 \|_3^2 \| \psi_1 \|_2^2,
	\end{equation}
	and we obtain with the Sobolev inequality that
	\begin{equation}
	\int \frac{dk}{|k|^2} \vert \langle \psi_1 - \psi_2, e^{-ik\, \cdot\, } \psi_1 \rangle \vert^2  \leq C \| \psi_1 - \psi_2 \|_{H^1( \mathbb{R}^3)}^2 \| \psi_1 \|_2^2.
	\end{equation}
	The second term of \eqref{eq:diff_sigma_1} can be bounded in a similar way, and we obtain the desired estimate. 
\end{proof}

We recall the definition of the resolvent
\begin{equation}
R_\varphi = q_{\psi_{\varphi}} \left( h_{\varphi} - e( \varphi ) \right)^{-1} q_{\psi_{\varphi}},
\end{equation}
where $q_{\psi_\varphi} = 1- \vert \psi_{\varphi} \rangle \langle \psi_\varphi \vert$. In the following Lemma we collect useful estimates on $R_\varphi$. 

\begin{lemma}
	\label{lemma:resolvent}
	There exists $C>0$ such that
	\begin{equation}
	\| R_{\varphi} \| = \Lambda ( \varphi )^{-1}, \quad \| \left( - \Delta + 1 \right)^{1/2} R_{\varphi}^{1/2} \| \leq C  ( 1 + \| \varphi \|_2 \| R_{\varphi}^{1/2} \| ) 
	\end{equation}
	for any $\varphi \in L^2( \mathbb{R}^3)$ with $e( \varphi ) < 0$.  
\end{lemma}

\begin{proof} 
	The first identity for the norm of the  resolvent follows immediately from the definition of the spectral gap $\Lambda ( \varphi )$ in \eqref{def:gap}. For $\psi \in L^2( \mathbb{R}^3)$ we have  
	\begin{equation}
	\| \left( - \Delta +1 \right)^{1/2}  R_\varphi^{1/2} \psi \|_2^2 = \langle \psi, \, R_\varphi^{1(2} \left( - \Delta +1 \right) R_\varphi^{1/2} \psi \rangle \,.
	\end{equation}
	It follows from Lemma \ref{lemma:potential} that  there exists $C>0$ such that 
	\begin{align}
	\| \left( - \Delta +1 \right)^{1/2} R_\varphi^{1/2} \psi \|_2^2  & \leq C \,  \langle \psi, R_{\varphi}^{1/2} \left( h_{\varphi} + C \| \varphi \|_2^2 \right) R_\varphi^{1/2} \psi \rangle \notag \\
	& = C\, \|q_{\psi_{\varphi}} \psi \|_2^2 + C \left(C \|  \varphi \|_2^2  + e( \varphi ) \right) \| R_\varphi^{1/2} \psi \|_2^2.
	\end{align}
	Since $e( \varphi ) <0$ this implies  
	the desired estimate. 
\end{proof}

\subsection{Perturbative properties of ground states and of the spectral gap }

Since the essential spectrum  of $h_{\varphi}$ is $\mathbb{R}_+$,  the assumption $e(\varphi)<0$ guarantees 
the existence of a ground state (denoted by $\psi_{\varphi}$) and of a spectral gap  $\Lambda(\varphi)>0$ of $h_{\varphi}$.
In the next two Lemmas we investigate the behavior of $\Lambda(\varphi)$ and $\psi_{\varphi}$ under $L^2$-perturbations of $\varphi$.

\begin{lemma}
	\label{lemma:specgap_pertub}
	Let $\varphi_0$  satisfy \eqref{eq:ass_old}, and  let $0 < \Lambda < \Lambda ( \varphi_0 )$. Then, there exists $\delta_\Lambda >0$ (depending, besides $\Lambda$, only on the spectrum of $h_{\varphi_0}$ and $\|\varphi_0\|_2$) such that 
	\begin{equation}
	\Lambda (\varphi ) \geq \Lambda \quad \mathrm{for \, all} \quad \varphi \in L^2( \mathbb{R}^3) \quad \mathrm{with} \quad \| \varphi - \varphi_0 \|_2 \leq \delta_\Lambda .
	\end{equation}
\end{lemma}

\begin{proof}
	By definition of the spectral gap
	\begin{equation}
	\Lambda ( \varphi ) = e_1 ( \varphi ) - e( \varphi ),
	\end{equation}
	where $e( \varphi )$ denotes the ground state energy of $h_\varphi$, and $e_1( \varphi )$ its first excited eigenvalue if it exists, or otherwise $e_1( \varphi) =0$ (which is the bottom of the essential spectrum). By the min-max principle we can write
	\begin{equation}
	e_1( \varphi ) = \inf_{\substack{A \subset L^2( \mathbb{R}^3) \\ \mathrm{dim} A =2}} \sup_{\substack{\psi \in A \\ \| \psi \|_2 =1}} \langle \psi, h_\varphi \psi \rangle .
	\end{equation}
	For $\psi \in H^1( \mathbb{R}^3)$ with $\| \psi \|_2 =1$ we find with Lemma \ref{lemma:potential}
	\begin{align}
	\langle \psi, h_\varphi \psi \rangle & = \langle \psi, h_{\varphi_0} \psi \rangle + \langle \psi, \, V_{\varphi - \varphi_0} \psi \rangle  \notag\\
	&\leq \langle \psi, \, h_{\varphi_0} \psi \rangle + C\| \varphi - \varphi_0 \|_2 \| \psi \|_{H^1( \mathbb{R}^3)}^2.
	\end{align}
	Moreover, for $\varepsilon >0$, 
	\begin{equation}
\| \psi \|_{H^1( \mathbb{R}^3)}^2 =  \langle \psi, \, h_{\varphi_0} \psi \rangle -  \langle \psi, \, V_{\varphi_0} \psi \rangle   + 1 \leq  \langle \psi, \, h_{\varphi_0} \psi \rangle + \varepsilon\| \psi \|_{H^1( \mathbb{R}^3)}^2 + C \varepsilon^{-1} \| \varphi_0 \|_2^2+1.
	\end{equation}
	Hence, choosing $\varepsilon = 1/2$, we find 
	\begin{equation}
	\| \psi \|_{H^1( \mathbb{R}^3)}^2 \leq 2 \langle \psi, \, h_{\varphi_0} \psi \rangle + C(  \| \varphi_0 \|_2^2 + 1 ) .
	\end{equation}
	Thus, if $\|\varphi-\varphi_0\|_2\leq \delta$, we have
	\begin{equation}
	\langle \psi, \, h_{\varphi} \psi \rangle \leq ( 1 + C \delta ) \langle \psi, \, h_{\varphi_0} \psi \rangle + C \delta  (\| \varphi_0 \|_2^2+1)  ,
	\end{equation}
	and similarly 
	\begin{equation}
	\langle \psi, \, h_{\varphi} \psi \rangle \geq ( 1 - C\delta ) \langle \psi, \, h_{\varphi_0} \psi \rangle - C \delta( \| \varphi_0 \|_2^2+1).
	\end{equation}
	Since $e( \varphi_0 ) ,  e( \varphi_1) \leq 0$, we therefore find 
	\begin{equation}
	\Lambda ( \varphi ) \geq \Lambda ( \varphi_0) -C\delta \left( e( \varphi_0) + e_1( \varphi_0) + 2 (\| \varphi_0 \|_2^2+1) \right) \geq \Lambda ( \varphi_0 ) - 2C \delta( \| \varphi_0 \|_2^2+1) > \Lambda
	\end{equation}
	for sufficiently small $\delta = \delta_\Lambda >0$. 
\end{proof}

\begin{lemma}
\label{lemma:gs_pertub}
Let $\varphi_0$  satisfy \eqref{eq:ass_old}, and  let $\varphi \in L^2( \mathbb{R}^3)$ with 
\begin{equation}
\label{ass:gs_perturb}
 \| \varphi - \varphi_0 \| \leq \delta_{\varphi_0}
\end{equation}
for sufficiently small $\delta_{\varphi_0} >0$. Then, there exists a unique positive and normalized ground state $\psi_{\varphi}$ of $h_\varphi$. Moreover, there exists $C>0$ (independent of $\varphi$) such that 
\begin{equation}
\label{eq:lemma_gs_perturb}
\| \psi_{\varphi_0} - \psi_{\varphi} \|_{H^1 ( \mathbb{R}^3)} \leq C  
\| \varphi - \varphi_0 \|_2 .
\end{equation}
\end{lemma}

\begin{proof}
We write
\begin{equation}
\psi_{\varphi} - \psi_{\varphi_0} = \int_0^1 d\mu \, \partial_\mu \psi_{\varphi_\mu},
\end{equation}
with $\varphi_\mu = \varphi_0 + \mu ( \varphi - \varphi_0 )$. Note that $\psi_{\varphi_\mu}$ is well defined for all $\mu\in [0,1]$, since 
\begin{equation}
\| \varphi_\mu - \varphi_0 \|_2 = \mu \| \varphi - \varphi_0 \|_2 \leq  \mu \delta_{\varphi_0} \leq  \delta_{\varphi_0}
\end{equation}
and therefore Lemma \ref{lemma:specgap_pertub} guarantees the existence of a spectral gap 
\begin{equation}
\label{eq:bound_Lambda_mu}
\Lambda ( \varphi_\mu ) \geq \Lambda > 0
\end{equation}
for sufficiently small $\delta_{\varphi_0}$, uniformly in $\mu\in [0,1]$. First order perturbation theory yields 
\begin{equation}
\partial_\mu \psi_{\varphi_\mu} = R_{\varphi_\mu }  V_{\varphi_0 - \varphi} \psi_{\varphi_\mu}
\end{equation} 
and it follows from Lemma \ref{lemma:potential} that 
\begin{align}
\| \psi_{\varphi_0} - \psi_{\varphi} \|_{H^1( \mathbb{R}^3)} & \leq  \int_0^1 d\mu \,  \| R_{\varphi_\mu} V_{\varphi - \varphi_0} \psi_{\varphi_\mu} \|_{H^1( \mathbb{R}^3)} \notag \\
&\leq C \int_0^1 d\mu \,  \| \left( - \Delta +1 \right)^{1/2}R_{\varphi_\mu}^{1/2}\|^2 \, \| \varphi - \varphi_0 \|_2  .
\end{align}
Lemma \ref{lemma:resolvent} shows that
\begin{equation}\label{st}
\| \left( - \Delta +1 \right)^{1/2}R_{\varphi_\mu}\| \leq C \left( 1 + \| \varphi_\mu \|_2 \| R_{\varphi_\mu} \| \right).
\end{equation}
Since $\| \varphi_\mu \|_2 \leq  \| \varphi_0 \|_2 + \mu \| \varphi - \varphi_0 \|_2 \leq \| \varphi_0 \|_2 + \delta_{\varphi_0}$, the bound \eqref{eq:bound_Lambda_mu} implies that the right-hand side of \eqref{st} is bounded independently of $\mu$. 
Hence the desired estimate \eqref{eq:lemma_gs_perturb} follows.
\end{proof}

\subsection{Pekar Functionals}
\label{sec:coerv}

Recall the definition of the Pekar Functionals $\Gfun$, $\Efun$ and $\Ffun$ in \eqref{def:G} and \eqref{def:EF}, and note that 
\begin{equation}
\Gfun(\psi,\varphi)=\Efun(\psi)+\|\varphi+\sigma_{\psi}\|_2^2 \,. 
\end{equation}
As was shown in \cite{L}, $\Efun$ admits a unique \emph{strictly positive and radially symmetric} minimizer, which is smooth and will be denoted by $\psi_{\mathrm{P}}$. Moreover, the set of all minimizers of $\Efun$ coincides with 
\begin{equation}\label{def:Theta}
\Theta(\psi_{\mathrm{P}})=\{e^{i\theta}\psi_{\mathrm{P}}(\, \cdot\, -y) \, | \, \theta\in[0,2\pi),\, y\in \Rthree\}.
\end{equation}
This clearly implies that the set of minimizers of $\Ffun$ coincides with 
\begin{equation}
\Omega(\varphi_{\mathrm{P}})=\{\varphi_{\mathrm{P}}(\, \cdot\, -y) \, | \, y\in \Rthree\} \quad \text{with} \quad \varphi_{\mathrm{P}}= -\sigma_{\psi_{\mathrm{P}}}.
\end{equation}

In the following we prove quadratic lower bounds for the Pekar Functionals $\Efun$ and $\Ffun$. The key ingredients  are the results obtained in \cite{lenzmann2009uniqueness}. In particular, these results allow to infer, using standard arguments, the following Lemma~\ref{lemma:EQuadBounds}, which provides the quadratic lower bounds for $\Efun$.  (We spell out  its proof for completeness in the Appendix; a very similar proof in a slightly different setting is also given in \cite{frank2019quantum}). Based on the bound for $\Efun$, it is then quite straightforward to obtain the  quadratic lower bound for $\Ffun$ in the subsequent Lemma~\ref{lemma:FQuadBounds}. 

\begin{lemma}[Quadratic Bounds for $\Efun$]
	\label{lemma:EQuadBounds}
	There exists a positive constant $\kappa$ such that, for any $L^2$-normalized $\psi \in H^1(\Rthree)$, 
	\begin{equation}\label{2.33}
	\Efun(\psi)-e_{\mathrm{P}}\geq\kappa \min_{y \in \Rthree \atop \theta\in[0,2\pi)} \|\psi-e^{i\theta} \psi_{\mathrm{P}}(\,\cdot\, -y)\|_{H^1(\Rthree)}^2=\dist_{H^1(\Rthree)}^2(\psi, \Theta(\psi_{\mathrm{P}})).
	\end{equation}
\end{lemma}

\begin{lemma}[Quadratic Bounds for $\Ffun$]
	\label{lemma:FQuadBounds}
	There exists a positive constant $\tau$ such that, for any $\varphi \in L^2(\Rthree)$, 
	\begin{equation}
	\Ffun(\varphi)-e_{\mathrm{P}}\geq \tau \min_{y\in \Rthree} \|\varphi-\varphi_{\mathrm{P}}(\,\cdot\, -y)\|_2^2=\tau \dist_{L^2(\Rthree)}^2(\varphi,\Omega(\varphi_{\mathrm{P}})).
	\end{equation}
\end{lemma}
\begin{proof}
	Recalling that 
	\begin{equation}
	\Ffun(\varphi)=\inf_{\|\psi\|_2=1\atop \psi \in H^1(\Rthree)} \Gfun(\psi,\varphi)
	\end{equation}
	our claim trivially follows by showing that for any $L^2$-normalized $\psi \in H^1(\Rthree)$ and $\varphi\in L^2(\R^3)$ 
	\begin{equation}
	\label{eq:Gestimate}
	\Gfun(\psi,\varphi)-e_{\mathrm{P}}\geq \tau\, \dist_{L^2(\Rthree)}^2(\varphi,\Omega(\varphi_{\mathrm{P}})).
	\end{equation}
	For any such $\psi$ let $y^*\in \Rthree$ and $\theta^* \in [0,2\pi)$ be such that
	\begin{equation}
	\|\psi- e^{i\theta^*}\psi_{\mathrm{P}}(\,\cdot\, -y^*)\|^2_{H^1(\Rthree)}=\dist_{H^1(\Rthree)}^2(\psi,\Theta(\psi_{\mathrm{P}})),
	\end{equation}
	and denote $e^{i\theta^*}\psi_{\mathrm{P}}(\,\cdot\, -y^*)$ by $\psi_{\mathrm{P}}^*$. By using the previous Lemma \ref{lemma:EQuadBounds}, the fact that $\psi$ and $\psi_{\mathrm{P}}^*$ are $L^2$-normalized, \eqref{eq:sigmaH1bounds} and completing the square, we obtain for, some positive $\kappa^*>0$,
	\begin{align}
	\Gfun(\psi,\varphi)-e_{\mathrm{P}}&=\Efun(\psi)-e_{\mathrm{P}}+\|\varphi+\sigma_{\psi}\|^2_2\geq \kappa\|\psi-\psi_{\mathrm{P}}^*\|_{H^1(\Rthree)}^2+\|\varphi+\sigma_{\psi}\|_2^2\notag\\
	&\geq \kappa^*\|\sigma_{\psi}-\sigma_{\psi_{\mathrm{P}}^*}\|_2^2+\|\varphi+\sigma_{\psi}\|_2^2\notag\\
	&=\|(1+\kappa^*)^{1/2} (\sigma_{\psi_{\mathrm{P}}^*}-\sigma_{\psi})-(1+\kappa^*)^{-1/2}(\varphi+\sigma_{\psi_{\mathrm{P}}^*})\|_2^2  + \frac{\kappa^*}{1+\kappa^*} \| \varphi+\sigma_{\psi_{\mathrm{P}}^*} \|_2^2 \notag\\
	&\geq\frac{\kappa^*}{1+\kappa^*} \|\varphi-\varphi_{\mathrm{P}}(\,\cdot\, -y^*)\|_2^2\geq \frac{\kappa^*}{1+\kappa^*} \dist_{L^2(\Rthree)}^2(\varphi,\Omega(\varphi_{\mathrm{P}})).
	\end{align}
	This completes the proof of \eqref{eq:Gestimate}, and hence of the Lemma, with $\tau= \kappa^*/(1+\kappa^*)$. 
\end{proof}

\begin{remark}
\label{remark:coercivityconseq}
The two previous quadratic bounds on $\Efun$ and $\Ffun$ clearly imply, together with \eqref{def:EF}, that, for any $L^2$-normalized $\psi\in H^1(\Rthree)$ and any $\varphi \in L^2(\Rthree)$, having low energy guarantees closeness to the surfaces of minimizers $\Theta(\psi_{\mathrm{P}})$ and $\Omega(\varphi_{\mathrm{P}})$, i.e.
\begin{equation}
\Gfun(\psi,\varphi)\leq e_{\mathrm{P}}+\varepsilon \quad \Rightarrow \quad \Efun(\psi),\Ffun(\varphi)\leq e_{\mathrm{P}}+\varepsilon \quad \Rightarrow \quad \dist_{H^1}^2(\psi,\Theta(\psi_{\mathrm{P}})), \dist_{L^2}^2(\varphi,\Omega(\varphi_{\mathrm{P}}))\leq C\varepsilon \,.
\end{equation}
\end{remark}

Finally, we exploit the previous estimate to obtain the following Lemma. It states  that for couples $(\psi,\varphi)$ which have low energy  $\psi$ must be close to $\psi_{\varphi}$, the ground state of $h_{\varphi}$, and $\varphi$ is close to $-\sigma_{\psi_{\varphi}}$, in the following sense.

\begin{lemma}
\label{lemma:help}
Let $\varepsilon>0$ be sufficiently small, $\psi\in H^1(\Rthree)$ be $L^2$-normalized, $\varphi \in L^2(\Rthree)$ and let $(\psi,\varphi)$ be such that 
\begin{equation}
\label{eq:Gfuncondition}
\Gfun(\psi,\varphi) \leq e_{\mathrm{P}}+\varepsilon \,.
\end{equation}
Then $h_\varphi$ has a positive ground state $\psi_{\varphi}$, and there exists $C>0$ (independent of $(\psi,\varphi))$ such that
\begin{align}
\label{eq:2.6.1.1}
\min_{\theta\in[0,2\pi)} \| \psi-e^{i\theta}\psi_{\varphi}\|_{H^1(\Rthree)}^2 &\leq C\varepsilon,\\
\label{eq:2.6.1.2}
\|\varphi+\sigma_{\psi_{\varphi}}\|^2_2 & \leq C \varepsilon.
\end{align}
\end{lemma}

\begin{proof}
Since  $\Ffun(\varphi)\leq \Gfun(\psi,\varphi)$ for any $L^2$-normalized $\psi \in H^1(\Rthree)$, Lemma \ref{lemma:FQuadBounds} implies that  for any $\delta>0$ there exists $\varepsilon_{\delta}>0$ such that $\dist_{L^2}(\varphi,\Omega(\varphi_{\mathrm{P}}))\leq \delta$ whenever $\Gfun(\psi,\varphi)\leq e_{\mathrm{P}}+\varepsilon_{\delta}$. Moreover, by Lemma \ref{lemma:specgap_pertub}, there  exists $\bar\delta>0$ such that  if $\dist_{L^2}(\varphi,\Omega(\varphi_{\mathrm{P}}))\leq \bar \delta$ then $\psi_{\varphi}$  exists. We then pick $\varepsilon=\varepsilon_{\bar \delta}$ and this guarantees that under the hypothesis of the Lemma $\psi_{\varphi}$ is well defined. 

Using Lemmas \ref{lemma:EQuadBounds} and \ref{lemma:FQuadBounds}, the assumption \eqref{eq:Gfuncondition} implies that there exist $y_1$ and $y_2$ such that
\begin{equation}
\label{eq:estimate1}
\min_{\theta\in [0.2\pi)}\|\psi-e^{i\theta}\psi_{\mathrm{P}}(\,\cdot\, -y_1)\|_{H^1(\Rthree)}^2\leq C\varepsilon,\quad \|\varphi-\varphi_{\mathrm{P}}(\, \cdot\, -y_2)\|_2^2\leq C \varepsilon.
\end{equation}
Moreover, since 
\begin{equation}
e_{\mathrm{P}}+\varepsilon\geq\Gfun(\psi,\varphi)=\Efun(\psi)+\|\varphi+\sigma_{\psi}\|_2^2\geq e_{\mathrm{P}}+\|\varphi+\sigma_{\psi}\|_2^2,
\end{equation} 
we also have
\begin{equation}
\label{eq:estimate2}
\|\varphi+\sigma_{\psi}\|_2^2\leq \varepsilon.
\end{equation}
In combination, the second bound in \eqref{eq:estimate1} and \eqref{eq:estimate2} imply
\begin{equation}
\label{eq:estimate3}
\|\varphi_{\mathrm{P}}(\,\cdot\, - y_2)+\sigma_{\psi}\|_2^2\leq C\varepsilon.
\end{equation}
Moreover, with the aid of  \eqref{eq:sigmaH1bounds} and the first bound in \eqref{eq:estimate1}, we obtain
\begin{equation}
\label{eq:estimate4}
\|\varphi_{\mathrm{P}}(\, \cdot\, -y_1)+\sigma_{\psi}\|_2^2=\|\sigma_{\psi_{\mathrm{P}}(\,\cdot\, -y_1)}-\sigma_{\psi}\|_2^2\leq C\min_{\theta\in [0,2\pi)}\|\psi-e^{i\theta}\psi_{\mathrm{P}}(\, \cdot\, -y_1)\|_{H^1}^2\leq C \varepsilon.
\end{equation}
By putting the second equation in \eqref{eq:estimate1}, \eqref{eq:estimate3} and \eqref{eq:estimate4} together, we can hence conclude that
\begin{equation}
\label{eq:estimate5}
\|\varphi-\varphi_{\mathrm{P}}(\, \cdot \, - y_1)\|_2\leq \|\varphi-\varphi_{\mathrm{P}}(\,\cdot\, -y_2)\|_2+\|\varphi_{\mathrm{P}}(\, \cdot\, -y_2)+\sigma_{\psi}\|_2+\|\sigma_{\psi}+\varphi_{\mathrm{P}}(\, \cdot\, -y_1)\|_2\leq C \varepsilon^{1/2}.
\end{equation}
Therefore, using Lemma \ref{lemma:gs_pertub}, we obtain
\begin{align}
\|\psi-e^{i\theta}\psi_{\varphi}\|_{H^1}&\leq \|\psi-e^{i\theta}\psi_{\mathrm{P}}(\,\cdot\, -y_1)\|_{H^1}+\|\psi_{\mathrm{P}}(\, \cdot\, -y_1)-\psi_{\varphi}\|_{H^1}  \notag \\ & = \|\psi-e^{i\theta}\psi_{\mathrm{P}}(\, \cdot\, -y_1)\|_{H^1}+\|\psi_{\varphi_{\mathrm{P}}(\, \cdot\, -y_1)}-\psi_{\varphi}\|_{H^1}\notag\\
&\leq \|\psi-e^{i\theta}\psi_{\mathrm{P}}(\, \cdot\, -y_1)\|_{H^1}+C \|\varphi_{\mathrm{P}}(\, \cdot\, -y_1)-\varphi\|_2  \,. 
\end{align}
This yields \eqref{eq:2.6.1.1} after taking the infimum over $\theta \in [0,2\pi)$ and using \eqref{eq:estimate5} and the first bound in \eqref{eq:estimate1}. To prove \eqref{eq:2.6.1.2}, we use \eqref{eq:estimate2}, \eqref{eq:sigmaH1bounds}, the normalization of $\psi$ and $\psi_{\varphi}$ and \eqref{eq:2.6.1.1} to obtain
\begin{equation}
\|\varphi+\sigma_{\psi_{\varphi}}\|_2\leq \|\varphi+\sigma_{\psi}\|_2+\|\sigma_{\psi}-\sigma_{\psi_{\varphi}}\|_2\leq \varepsilon^{1/2}+C\min_{\theta\in [0,2\pi)}\|\psi-e^{i\theta}\psi_{\varphi}\|_{H^1}\leq C\varepsilon^{1/2}.
\end{equation}
\end{proof}

\section{Proof of the Main Results }
\label{section3}

The conservation of $\Gfun$ along solutions of the Landau--Pekar equations 
allows to apply the tools developed in Section \ref{section2} to get results valid for all times. This will in particular allow us to prove the results stated in Section \ref{section1}.
When combined with energy conservation, Remark \ref{remark:coercivityconseq} shows that we can estimate the distance to the sets of Pekar minimizers of solutions of the Landau--Pekar equations only in terms of the energy of their initial data. Since $\Omega ( \varphi_{\mathrm{P}})$ contains only real-valued functions this  yields bounds on the $L^2$-norm of the imaginary part of $\varphi_t$.  
That is, there exists a $C>0$ such that if $( \psi_t, \varphi_t)$ solves the Landau--Pekar equations \eqref{eq:LP} with initial data $(\psi_0,\varphi_0)$, then  
\begin{align}\label{rmk:Im,Re}
\min_{\substack{y \in \mathbb{R}^3 \\ \theta \in [0, 2 \pi )}}\| \psi_t - e^{i \theta} \psi_{\mathrm{P}}(\,\cdot\, -y) \|&_{H^1( \mathbb{R}^3)}^2 \leq C (\Gfun(\psi_0,\varphi_0)-e_{\mathrm{P}}), \quad \| \Im\, \varphi_t \|_2^2 \leq C (\Gfun(\psi_0,\varphi_0)-e_{\mathrm{P}}), \nonumber\\
& \min_{y \in \mathbb{R}^3}\| \Re\,\varphi_t - \varphi_{\mathrm{P}}(\,\cdot\,-y) \|_2^2 \leq C (\Gfun(\psi_0,\varphi_0)-e_{\mathrm{P}})
\end{align}
for all $t\in \mathbb{R}$ and $\alpha>0$. It is then straightforward to obtain a proof of Theorem \ref{thm:specgap}.
 
\begin{proof}[Proof of Theorem \ref{thm:specgap}]
Let $0<\Lambda<\Lambda(\varphi_{\mathrm{P}})$ and let $( \psi_t, \varphi_t)$ denote the solution to the Landau--Pekar equations with initial data $( \psi_0, \varphi_0 )$ satisfying 
$
\Gfun(\psi_0,\varphi_0)\leq e_{\mathrm{P}}+\varepsilon_{\Lambda}
$. 
From~\eqref{rmk:Im,Re} we deduce that for any $t\in \R$ there exists $y_t \in \Rthree$ such that
\begin{equation}
\|\varphi_t-\varphi_{\mathrm{P}}(\,\cdot\, -y_t)\|_2^2\leq C\varepsilon_{\Lambda}
\end{equation}
for some $C>0$. 
Since the spectrum of $h_{\varphi_{\mathrm{P}}(\,\cdot\, -y)}$ and $\|\varphi_{\mathrm{P}}(\,\cdot\, -y)\|_2$ are independent of $y\in \Rthree$, Theorem \ref{thm:specgap} now follows immediately from Lemma \ref{lemma:specgap_pertub} by taking $\varepsilon_{\Lambda}= C^{-1}\delta_{\Lambda}^2$, where $\delta_{\Lambda}$ is the same as in Lemma \ref{lemma:specgap_pertub}.
\end{proof}

Conservation of energy also allows to extend the validity of Lemma \ref{lemma:help} for all times. If $(\psi_t,\varphi_t)$ solves \eqref{eq:LP}  with initial data $(\psi_0,\varphi_0)$ satisfying $\Gfun(\psi_0,\varphi_0)\leq e_{\mathrm{P}}+\varepsilon$ for a sufficiently small $\varepsilon$, then $\psi_{\varphi_t}$ is well defined for all times and 
\begin{equation}\label{rmk:psiphi}
\min_{\theta\in[0,2\pi)}\| \psi_t-e^{i\theta}\psi_{\varphi_t}\|_{H^1(\Rthree)}^2\leq C\varepsilon, \quad
\|\varphi_t+\sigma_{\psi_{\varphi_t}}\|^2_2\leq C \varepsilon.
\end{equation}
Moreover, Theorem \ref{thm:specgap} implies that for all times $\Lambda(\varphi_t)\geq \Lambda$ for a suitable $\Lambda>0$. It thus follows from Lemmas~\ref{lemma:wellposedness} and~\ref{lemma:resolvent} that for some $C>0$
\begin{equation}\label{3.4}
\| R_{\varphi_t} \| \leq C \quad \mathrm{and} \quad \| ( - \Delta + 1)^{1/2} R_{\varphi_t}^{1/2} \| \leq C  \quad \text{for  all} \quad t \in \mathbb{R},
\end{equation}
where as above $R_{\varphi_t} = q_t \left( h_{\varphi_t} - e( \varphi_t ) \right)^{-1} q_t$ and $q_t = 1- p_t = 1- \vert \psi_{\varphi_t} \rangle \langle \psi_{\varphi_t} \vert $.

With these preparations, we are now ready to prove Corollary 
\ref{cor:adiab}.

\begin{proof}[Proof of Corollary \ref{cor:adiab}] The proof follows closely the ideas of the proof of \cite[Theorem II.1]{LRSS}, hence we allow ourselves to be a bit sketchy at some points and  refer to \cite{LRSS} for more details.
It follows from the Landau--Pekar equations \eqref{eq:LP}  that
\begin{equation}\label{3.5}
\alpha^2 \partial_t V_{\varphi_t}  =  V_{\Im\, \varphi_t}, \quad \alpha^2 \partial_t V_{\Im\, \varphi_t} = - V_{ \Re\, \varphi_t + \sigma_{\psi_t}}.
\end{equation}
Lemmas~\ref{lemma:wellposedness}--\ref{lemma:resolvent} imply, together with \eqref{rmk:Im,Re}, that there exists $C>0$ such that 
\begin{equation}\label{3.6}
\| R_{\varphi_t} V_{\Im\, \varphi_t} \|^2 \leq C \varepsilon \quad \mathrm{for \, all } \quad t \in \mathbb{R}.
\end{equation}
In the same way, by the triangle inequality, Lemma \ref{lemma:potential} and  \eqref{rmk:psiphi}, there exists $C>0$ such that 
\begin{equation}\label{3.7}
\| R_{\varphi_t} V_{\Re\, \varphi_t + \sigma_{\psi_t}} \|^2 \leq C \min_{\theta \in (0, 2 \pi ]}\| \psi_t - e^{i \theta } \psi_{\varphi_t} \|_{H^1( \mathbb{R}^3)}^2 + C \| \Re\, \varphi_t + \sigma_{\psi_{\varphi_t}} \|_2^2  \leq C \varepsilon   \quad \mathrm{for \, all } \quad t \in \mathbb{R}.
\end{equation}
Moreover, it follows from
\begin{equation}\label{3.8}
\alpha^2 \partial_t \psi_{\varphi_t} = - R_{\varphi_t} V_{\Im \, \varphi_t} \psi_{\varphi_t}
\end{equation}
that 
\begin{equation}
\label{eq:deriv_R}
\alpha^2 \partial_t R_{\varphi_t} =  p_t V_{\Im\, \varphi_t} R_{\varphi_t}^2 + R_{\varphi_t}^2 V_{\Im\, \varphi_t} p_t - R_{\varphi_t}  \left( V_{\Im\, \varphi_t} - \langle \psi_{\varphi_t}, V_{\Im\, \varphi_t} \psi_{\varphi_t} \rangle \right) R_{\varphi_t} 
\end{equation}
(see \cite[Lemma IV.2]{LRSS}) and by the same arguments as above that
\begin{equation}\label{3.9}
\| \left( - \Delta + 1\right)^{1/2} \partial_t R_{\varphi_t} \left( - \Delta + 1 \right)^{1/2} \| \leq C\varepsilon^{1/2} \alpha^{-2} \quad \text{for all} \quad t \in \mathbb{R}. 
\end{equation}

Recall the definitions of $\widetilde{\psi}_t$  and $\nu$ in \eqref{eq:phase_add}. 
The same computations as in \cite[Eqs.~(58)--(65)]{LRSS}, using 
\begin{equation}
\label{eq:PI}
q_t \, e^{i \int_0^t ds \, e( \varphi_s )} \psi_t = i\, R_{\varphi_t} \, \partial_t \,e^{i \int_0^t ds \, e( \varphi_s )} \psi_{t}
\end{equation}
and integration by parts, lead to 
\begin{subequations}\label{3.11}
\begin{align}
\| \widetilde{\psi}_t - \psi_{\varphi_t} \|_2^2 &=  2 \alpha^{-2} \Im\,  \langle \widetilde{\psi}_t, \, R_{\varphi_t}^2 V_{\Im\, \varphi_t} \psi_{\varphi_t} \rangle\label{eq:boundary} \\
&\quad +  2 \alpha^{-2} \int_0^t ds \, \nu (s) \,  \Re \langle  \widetilde{\psi}_s, \, R_{\varphi_s}^2 V_{\Im\, \varphi_s} \psi_{\varphi_s} \rangle \label{eq:Phase}\\
&\quad + 2 \alpha^{-4} \int_0^t ds \, \Im \langle \widetilde{\psi}_s, \, R_{\varphi_s} \left( R_{\varphi_s} V_{\Im\, \varphi_s}  \right)^2 \psi_{\varphi_s} \rangle \label{eq:deriv_gs}\\
&\quad + 2 \alpha^{-4} \int_0^t ds \, \Im \langle \widetilde{\psi}_s, \, R_{\varphi_s}^2 V_{\Re\, \varphi_s + \sigma_{\psi_s}}  \psi_{\varphi_s} \rangle\label{eq:deriv_pot} \\
&\quad - 2 \alpha^{-2} \int_0^t ds \,\left(  \Im \langle \widetilde{\psi}_s, \, \left( \partial_s R_{\varphi_s}^2 \right) V_{\Im\, \varphi_s} \psi_{\varphi_s}\rangle + \alpha^2  \nu (s)  \, \Im \langle \widetilde{\psi}_s,\,  \psi_{\varphi_s} \rangle \right) \label{eq:deriv_R^2} .
\end{align}
\end{subequations}
The difference to the calculations in \cite{LRSS} are the additional terms \eqref{eq:Phase} and the second term in \eqref{eq:deriv_R^2} resulting from the phase $\nu$. While \eqref{eq:Phase} is, as we show below, only  a subleading error term, the phase  in \eqref{eq:deriv_R^2} leads to a crucial cancellation. This cancellation allows to integrate by parts once more, and finally results in the improved estimate in Corollary \ref{cor:adiab}. 

We shall now estimate the various terms in \eqref{3.11}. Since $\| q_t \widetilde{\psi}_t \|_2 \leq  \| \widetilde{\psi}_t  - \psi_{\varphi_t}\|_2$,  we find for the first term using  \eqref{3.4} and \eqref{3.6} 
\begin{equation}\label{a}
\vert \eqref{eq:boundary} \vert \leq C \alpha^{-2} \eps^{1/2} \| \widetilde{\psi}_t - \psi_{\varphi_t} \|_2 \leq \delta \| \widetilde{\psi}_t - \psi_{\varphi_t} \|_2^2 + C \delta^{-1} \alpha^{-4} \varepsilon 
\end{equation}
for arbitrary $\delta >0$. Moreover, we have $\vert \nu (s) \vert \leq C \alpha^{-4} \varepsilon $ for all $s\in \R$, and   find for the second term
\begin{equation}\label{b}
\vert  \eqref{eq:Phase} \vert \leq C \alpha^{-6} \eps^{3/2} \int_0^t ds\, \| \widetilde{\psi}_s - \psi_{\varphi_s} \|_2 \,. 
\end{equation}
For the third term, we integrate by parts using \eqref{eq:PI} once more, with the result that
\begin{align}
\eqref{eq:deriv_gs} &= -2 \alpha^{-4}\,  \Re\, \langle \widetilde{\psi}_t, \, R_{\varphi_t}^2 \left( R_{\varphi_t} V_{\Im\, \varphi_t} \right)^2 \psi_{\varphi_t} \rangle  + 2 \alpha^{-4} \int_0^t ds \, \nu (s) \, \Im\, \langle \widetilde{\psi}_s, R_{\varphi_s}^2 \left( R_{\varphi_s} V_{\Im\, \varphi_s} \right)^2 \psi_{\varphi_s } \rangle \notag \\
&\quad + 2 \alpha^{-4} \int_0^t ds \, \Re\, \langle \widetilde{\psi}_s, \, \partial_s \left( R_{\varphi_s}^2 \left( R_{\varphi_s} V_{\Im\, \varphi_s} \right)^2 \psi_{\varphi_s}  \right)  \rangle .
\end{align}
The first two terms can be bounded in the same way as \eqref{eq:boundary} and \eqref{eq:Phase}. For the third term, note that the r.h.s. of the inner product depends on time $s$ through $\varphi_s$ only, hence its time derivative leads to another factor of $\alpha^{-2}$. With \eqref{3.5}, \eqref{3.8} and \eqref{eq:deriv_R} we compute its time derivative. From the time derivative of the resolvent in  \eqref{eq:deriv_R}, we obtain one term for which the projection $p_s$ hits $\widetilde{\psi}_s$ on the l.h.s. of the inner product, in which case we can only bound $\| p_s \widetilde{\psi}_s \|_2 \leq  1$. For the remaining terms, we use  $\| q_s \widetilde{\psi}_s \|_2 \leq \| \widetilde{\psi}_s - \psi_{\varphi_s} \|_2$ instead. With the same arguments as above and \eqref{3.7},  we obtain
\begin{equation}\label{c}
\vert \eqref{eq:deriv_gs} \vert \leq \delta \| \widetilde{\psi}_t - \psi_{\varphi_t} \|_2^2  + C \delta^{-1} \alpha^{-8} \varepsilon^2 +C  \alpha^{-6} \eps \int_0^t ds \, \| \widetilde{\psi}_s - \psi_{\varphi_s} \|_2  + C \alpha^{-6} \eps^{3/2} |t| 
\end{equation}
for any $\delta>0$. 
For the forth term \eqref{eq:deriv_pot}, we first split 
\begin{equation}
\label{eq:split_pot}
\eqref{eq:deriv_pot} =
 2 \alpha^{-4} \int_0^t ds \, \left( \Im\, \langle \widetilde{\psi}_s, \, R_{\varphi_s}^2 V_{\sigma_{\psi_s}- \sigma_{\psi_{\varphi_s}}  }  \psi_{\varphi_s} \rangle +  \Im\, \langle \widetilde{\psi}_s, \, R_{\varphi_s}^2 V_{  \Re\,\varphi_s +  \sigma_{\psi_{\varphi_s}} }  \psi_{\varphi_s} \rangle \right)\,.
\end{equation}
Lemmas~\ref{lemma:wellposedness}--\ref{lemma:resolvent} and \eqref{3.4} imply that we can bound $\|R_{\varphi_s}^2 V_{\sigma_{\psi_s}- \sigma_{\psi_{\varphi_s}}  } \|\leq C \| \widetilde \psi_s - \psi_{\varphi_s} \|_2$  in the first term. For the second term, we observe that the r.h.s. of the inner product depends  on $s$ again only through $\varphi_s$, whose time derivative is of order $\alpha^{-2}$. We thus again use \eqref{eq:PI} and integration by parts, and proceed as above. For the calculation, we need to bound the time derivative of $\sigma_{\psi_{\varphi_s}}$, which can be done with the aid \cite[Lemma~II.4]{LMRSS}, with the result that $\| \partial_s \sigma_{\psi_{\varphi_s}} \|_2\leq C \eps^{1/2} \alpha^{-2}$. Altogether, this shows  
that
\begin{align}\label{d}
\vert \eqref{eq:deriv_pot} \vert  & \leq  C  \alpha^{-4} \int_0^t ds \, \| \widetilde{\psi}_s - \psi_{\varphi_s} \|_2^2 + \delta \| \widetilde{\psi}_t - \psi_{\varphi_t} \|_2^2  + C \delta^{-1} \alpha^{-8} \varepsilon \notag \\
&\quad +C  \alpha^{-6} \eps^{1/2} \int_0^t ds \, \| \widetilde{\psi}_s - \psi_{\varphi_s} \|_2 + C \alpha^{-6} \eps |t| 
\end{align}
for any $\delta>0$. 
For the last term, we compute using \eqref{eq:deriv_R} 
\begin{align}
\eqref{eq:deriv_R^2} & = -6 \alpha^{-4} \int_0^t ds \, \Im \langle \widetilde{\psi}_s, \,  R_{\varphi_s}^3 V_{\Im\, \varphi_s} p_s  V_{\Im\, \varphi_s} \psi_{\varphi_s} \rangle\notag \\
&\quad + 2 \alpha^{-4} \int_0^t ds \, \Im \langle \widetilde{\psi}_s, \, \left(  R_{\varphi_s}^2 V_{\Im\, \varphi_s} R_{\varphi_s} + R_{\varphi_s} V_{\Im\, \varphi_s} R_{\varphi_s}^2 \right) V_{\Im\, \varphi_s} \psi_{\varphi_s} \rangle . \label{eq:deriv_R^2_spelledout}
\end{align}
Note that the phase $\nu(s)$ cancels the contribution of $\partial_s R_{\varphi_s} $ projecting onto $\psi_{\varphi_s}$ (the first term of \eqref{eq:deriv_R}). This cancellation is important, since the integration by parts argument using \eqref{eq:PI} would not be applicable to this term. It can be applied to all the terms in \eqref{eq:deriv_R^2_spelledout}, however, proceeding as above, with the result that
\begin{equation}\label{e}
| \eqref{eq:deriv_R^2}| \leq  \delta \| \widetilde{\psi}_t - \psi_{\varphi_t} \|_2^2  + C \delta^{-1} \alpha^{-8} \varepsilon^2 +C  \alpha^{-6}\eps \int_0^t ds \, \| \widetilde{\psi}_s - \psi_{\varphi_s} \|_2 + C \alpha^{-6} \eps^{3/2} |t| 
\end{equation}
for any $\delta>0$. 

Collecting the bounds in \eqref{a}, \eqref{b}, \eqref{c}, \eqref{d} and \eqref{e},  Eq.~\eqref{3.11} shows that
\begin{align}\nonumber
\| \widetilde{\psi}_t - \psi_{\varphi_t} \|_2^2  & \leq C  \alpha^{-4} \varepsilon  + C \alpha^{-6} \eps^{1/2} \int_0^t ds\, \| \widetilde{\psi}_s - \psi_{\varphi_s} \|_2 + C  \alpha^{-4} \int_0^t ds \, \| \widetilde{\psi}_s - \psi_{\varphi_s} \|_2^2+ C \alpha^{-6} \eps |t| \\  &\leq C  \alpha^{-4} \varepsilon  + C  \alpha^{-4} \int_0^t ds \, \| \widetilde{\psi}_s - \psi_{\varphi_s} \|_2^2  + C \alpha^{-6} \eps |t|  
\end{align}
for $\alpha\gtrsim 1$ and $\eps\lesssim 1$. 
A Gronwall type argument finally yields the desired bound \eqref{eq:adiab_comb}. 
\end{proof}

\appendix

\section{Appendix: Proof of Lemma \ref{lemma:EQuadBounds}}
\label{appendixA}

In this appendix we give the proof of Lemma \ref{lemma:EQuadBounds}. As already mentioned, the result follows from 
the work in \cite{lenzmann2009uniqueness} by standard arguments. We follow closely the proof  given in \cite{frank2019quantum} of a corresponding result in the slightly different setting of a confined polaron. 

\begin{proof}[Proof of Lemma \ref{lemma:EQuadBounds}]
	\textbf{Step 1}: For any $L^2$-normalized $\psi \in H^1(\Rthree)$, there  exists $\bar\theta\in [0,2\pi)$ and $\bar y\in \Rthree$ such that
	\begin{equation}
	\|e^{i\bar\theta}\psi(\, \cdot\, -\bar y)-\psi_{\mathrm{P}}\|_2=\min_{y, \theta} \|e^{i\theta}\psi(\, \cdot\, - y)-\psi_{\mathrm{P}}\|_2.
	\end{equation}
	By invariance of $\Efun$ under translations and changes of phase, it is then sufficient to show that for any $L^2$-normalized $\psi$ such that
	\begin{equation}
	\label{eq:conditions}
	\|\psi-\psi_{\mathrm{P}}\|_2=\min_{y, \theta} \|\psi-e^{i\theta}\psi_{\mathrm{P}}(\, \cdot \, -y)\|_2, 
	\end{equation}
	the inequality
	\begin{equation}
	\label{eq:goal}
	\Efun(\psi)-e_{\mathrm{P}}\geq \kappa \|\psi-\psi_{\mathrm{P}}\|_{H^1(\Rthree)}^2
	\end{equation}
	holds (for some $\kappa>0$ independent of $\psi$). In fact, this is stronger than the desired bound \eqref{2.33}. 
	We henceforth only work with $L^2$-normalized $\psi$ satisfying \eqref{eq:conditions}, and denote $\delta=\psi-\psi_{\mathrm{P}}$. Observe that any $\psi$ satisfying \eqref{eq:conditions} also satisfies
	\begin{equation}
	\bra{\psi}\ket{\psi_{\mathrm{P}}}\geq 0, \quad \bra{\psi}\ket{\partial_i\psi_{\mathrm{P}}}=0 \  \text{for} \ i=1,2,3.
	\end{equation}
	
	\textbf{Step 2}: We first prove the quadratic lower bound \eqref{eq:goal} locally around $\psi_{\mathrm{P}}$ for any $L^2$-normalized $\psi$ satisfying \eqref{eq:conditions}. By straightforward computations, using that 
	\begin{equation}
	\label{eq:deltaprop}
	\|\delta\|_2^2=2-2 \bra{\psi_{\mathrm{P}}}\ket{\psi}=-2\bra{\psi_{\mathrm{P}}}\ket{\delta}
	\end{equation}
	since both $\psi_{\mathrm{P}}$ and $\psi$ are $L^2$-normalized, 
	we obtain
	\begin{equation}	\label{eq:Efunexp}
	\Efun(\psi)-e_{\mathrm{P}}=\text{Hess}_{\psi_{\mathrm{P}}}(\delta)+O(\|\delta\|^3_{H^1(\Rthree)}), 
	\end{equation}
	with
	\begin{align}
	\text{Hess}_{\psi_{\mathrm{P}}}(\delta) & =\expval{QL_{-}Q}{\Im\, \delta}+\expval{QL_{+}Q}{\Re\, \delta},\notag\\
	Q & =1-\ket{\psi_{\mathrm{P}}}\bra{\psi_{\mathrm{P}}},\notag \\
	L_{-} &=h_{\varphi_{\mathrm{P}}}  -e(\varphi_{\mathrm{P}}),\notag\\
	L_{+}& =L_{-}-4X, \notag\\
	X & =  (2\pi)^3 \psi_{\mathrm{P}} (-\Delta)^{-1} \psi_{\mathrm{P}} \,,
	\end{align}
	where in the last formula for $X$, $\psi_{\mathrm{P}}$ has to be understood as a multiplication operator.
	
	The Euler--Lagrange equation for the minimization of $\Efun$ reads $L_{-} \psi_{\mathrm{P}}=0$, and since $L_{-}$ is a Schr\"odinger operator and $\psi_{\mathrm{P}}$ is strictly positive, $L_{-}$ has   $0$ as its lowest eigenvalue, and a gap above.  Therefore we have
	\begin{equation}
	QL_{-}Q\geq \kappa_1 Q
	\end{equation}
	for some $\kappa_1>0$. 
	Moreover, it was shown in \cite{lenzmann2009uniqueness} that the kernel of $L_{+}$ coincides with $\text{span}_{i=1,2,3}\{\partial_i  \psi_{\mathrm{P}}\}$ and from this we can infer the existence of a $\kappa_2 > 0$ such that
	\begin{equation}
	QL_{+}Q\geq \kappa_{2} Q' \quad \text{with} \ Q'=Q-\sum_{i=1}^3 \|\partial_i \psi_{\mathrm{P}}\|_2^{-2} \ket{\partial_i \psi_{\mathrm{P}}}\bra{\partial_i \psi_{\mathrm{P}}}.
	\end{equation}
	Recall that $Q' \delta= Q \delta$ by assumption on $\psi$ and orthogonality of $\psi_{\mathrm{P}}$ to its partial derivatives. With 
	$\kappa' =\min\{\kappa_1,\kappa_2\}$ we thus have
	\begin{align}
	\text{Hess}_{\psi_{\mathrm{P}}}(\delta)\geq \kappa_1\|Q\Im\, \delta\|_2^2+\kappa_2\|Q' \Re\,  \delta\|_2^2\geq \kappa' \|Q\delta\|^2_2.
	\end{align}
	Using again \eqref{eq:deltaprop} we see that
	\begin{align}
	\|Q\delta\|_2^2=\|\delta\|_2^2-\bra{\psi_{\mathrm{P}}}\ket{\delta}^2=\|\delta\|_2^2\left(1-\frac 1 4 \|\delta\|_2^2\right)\geq \frac 1 2 \|\delta\|_2^2,
	\end{align}
	which finally implies that 
	\begin{align}
	\label{eq:L2hess}
	\text{Hess}_{\psi_{\mathrm{P}}}(\delta)\geq \frac {\kappa '} 2 \|\delta\|_2^2.
	\end{align}
	
	We now want to improve this bound to include the full $H^1$-norm of $\delta$. Using the regularity of $\psi_{\mathrm{P}}$ it is rather straightforward to show that
	\begin{align}
	L_{-}   =QL_{-}Q &\geq -\Delta-C \,, \notag\\
	QL_{+}Q & \geq -\Delta -C
	\end{align}
	which implies, that
	\begin{equation}
	\label{eq:H1hess}
	\text{Hess}_{\psi_{\mathrm{P}}}(\delta)\geq \|\delta\|_{H^1}^2-C\|\delta\|_2^2.
	\end{equation}
	By interpolating between \eqref{eq:L2hess} and \eqref{eq:H1hess}, we finally obtain
	\begin{equation}
	\text{Hess}_{\psi_{\mathrm{P}}}(\delta)\geq \frac {\kappa'}{\kappa'+2C}\|\delta\|_{H^1}^2=\kappa''\|\delta\|_{H^1}^2.
	\end{equation}
	In combination with \eqref{eq:Efunexp}, we conclude that
	\begin{equation}\label{A.16}
	\Efun(\psi)-e_{\mathrm{P}}\geq \kappa''\|\delta\|_{H^1}^2-C\|\delta\|_{H^1}^3
	\end{equation}
	for any $L^2$-normalized $\psi$ satisfying \eqref{eq:conditions}, 
	which shows that \eqref{eq:goal} holds  for $\|\delta\|_{H^1}$ sufficiently small.
		
	\textbf{Step 3}: We now extend the previous local bound to show that \eqref{eq:goal} holds globally. Suppose by contradiction that there does not exist a universal $\kappa$ such that \eqref{eq:goal} holds.  Then there exists a sequence $\psi_n$ of $L^2$-normalized functions satisfying \eqref{eq:conditions} such that
	\begin{align}
	\label{eq:boundsequence}
	\Efun(\psi_n)\leq e_{\mathrm{P}}+\frac 1 n \|\psi_n-\psi_{\mathrm{P}}\|_{H^1}^2\leq \frac 2 n \|\psi_n\|_{H^1}^2 + C\,.
	\end{align}
	One readily checks that 
	\begin{align}
	\Efun(\psi_n)\geq \frac 1 2 \|\psi_n\|_{H^1}^2-C \,, 
	\end{align}
	hence $\psi_n$ must be bounded in $H^1(\Rthree)$.  Again using \eqref{eq:boundsequence}, we conclude that $\psi_n$ must be a minimizing sequence for $\Efun$. It was proven in \cite{L} that any minimizing sequence converges in $H^1(\Rthree)$ to a minimizer of $\Efun$, i.e., an element of $\Theta(\psi_{\mathrm{P}})$ in \eqref{def:Theta}, and since $\psi_n$ satisfies \eqref{eq:conditions} this implies that $\psi_n\xrightarrow{H^1} \psi_{\mathrm{P}}$. This yields a contradiction, since we already know by \eqref{A.16} that locally the bound \eqref{eq:goal} holds. 
\end{proof}

\paragraph{Acknowledgments.} 
Funding from the European Union's Horizon 2020 research and innovation programme under the ERC grant agreement No 694227 (D.F. and R.S.) and under the Marie Sk\l{}odowska-Curie  Grant Agreement No. 754411 (S.R.) is gratefully acknowledged.

\vspace{0.5cm}

\noindent
(Dario Feliciangeli) Institute of Science and Technology Austria (IST Austria)\\ Am Campus 1, 3400 Klosterneuburg, Austria\\ 
E-mail address: \texttt{dario.feliciangeli@ist.ac.at} \\

\noindent
(Simone Rademacher) Institute of Science and Technology Austria (IST Austria)\\ Am Campus 1, 3400 Klosterneuburg, Austria\\ 
E-mail address: \texttt{simone.rademacher@ist.ac.at} \\

\noindent
(Robert Seiringer) Institute of Science and Technology Austria (IST Austria)\\ Am Campus 1, 3400 Klosterneuburg, Austria\\ 
E-mail address: \texttt{robert.seiringer@ist.ac.at}

\end{document}